\documentclass[10pt,journal,epsfig]{IEEEtran}

\usepackage[dvips]{graphicx}
\usepackage{graphicx}
\usepackage{amssymb}
\usepackage{cite}
\usepackage{subfigure}
\usepackage{amsmath}

\begin{document}

\title{One-Bit Quantization Design and Adaptive Methods for Compressed Sensing }

\author{Jun Fang, Yanning Shen, and Hongbin Li,~\IEEEmembership{Senior
Member,~IEEE}
\thanks{Jun Fang, and Yanning Shen are with the National Key Laboratory on Communications,
University of Electronic Science and Technology of China, Chengdu
611731, China, Emails: JunFang@uestc.edu.cn,
201121260110@std.uestc.edu.cn}
\thanks{Hongbin Li is with the Department of Electrical and
Computer Engineering, Stevens Institute of Technology, Hoboken, NJ
07030, USA, E-mail: Hongbin.Li@stevens.edu}
%\thanks{Huiping Duan is with the School of Electronic Engineering,
%University of Electronic Science and Technology of China, Chengdu
%611731, China, Email: huipingduan@uestc.edu.cn}
\thanks{This work was supported in part by the National
Science Foundation of China under Grant 61172114, and the National
Science Foundation under Grant ECCS-0901066.}}

\maketitle

%that can largely improve the performance compared with
%conventional $1$-bit compressive sensing without considering the
%influence of threshold

%Moreover, sufficient condition which guarantee perfect recovery is
%derived when proper thresholds are set. Furthermore,

%with adaptive thresholds based on multi-bit quantization
%framework, which consists of iteratively estimating the the sparse
%signal and adjusting quantizer based on estimates.

\begin{abstract}
There have been a number of studies on sparse signal recovery from
one-bit quantized measurements. Nevertheless, little attention has
been paid to the choice of the quantization thresholds and its
impact on the signal recovery performance. This paper examines the
problem of one-bit quantizer design for sparse signal recovery.
Our analysis shows that the magnitude ambiguity that ever plagues
conventional one-bit compressed sensing methods can be resolved,
and an arbitrarily small reconstruction error can be achieved by
setting the quantization thresholds close enough to the original
unquantized measurements. Note that the unquantized data are
inaccessible by us. To overcome this difficulty, we propose an
adaptive quantization method that iteratively refines the
quantization thresholds based on previous estimate of the sparse
signal. Numerical results are provided to collaborate our
theoretical results and to illustrate the effectiveness of the
proposed algorithm.
\end{abstract}

%in a way such that the thresholds converges to the desired values

%framework for sparse signal reconstruction from one-bit quantized
%measurements.

%Also, a scenario where the data is sent through a binary symmetric
%channel (BSC) is studied.

%In addition, the magnitude ambiguity issue that ever plagues
%conventional one-bit compressed sensing methods can be resolved.


\section{Introduction}
Compressive sensing is a recently emerged paradigm of signal
sampling and reconstruction, the main purpose of which is to
recover sparse signals from much fewer linear measurements
\cite{CandesTao05,Donoho06}
\begin{align}
\boldsymbol{y}=\boldsymbol{Ax}
\end{align}
where $\boldsymbol{A}\in\mathbb{R}^{m\times n}$ is the sampling
matrix with $m\ll n$, and $\boldsymbol{x}$ denotes the
$n$-dimensional sparse signal with only $K$ nonzero coefficients.
Such a problem has been extensively studied and a variety of
algorithms that provide consistent recovery performance guarantee
were proposed, e.g.
\cite{CandesTao05,Donoho06,TroppGilbert07,Wainwright09}.
Conventional compressed sensing assumes infinite precision of the
acquired measurements. In practice, however, signals need to be
quantized before further processing, that is, the real-valued
measurements need to be mapped to discrete values over some finite
range. Besides, in some sensing systems (e.g. distributed sensor
networks), data acquisition is expensive due to limited bandwidth
and energy constraints \cite{AkyildizSu02}. Aggressive
quantization strategies which compress real-valued measurements
into one or only a few bits of data are preferred in such
scenarios. Another benefit brought by low-rate quantization is
that it can significantly reduce the hardware complexity and cost
of the analog-to-digital converter (ADC). As pointed out in
\cite{LaskaWen11}, low-rate quantizer can operate at a much higher
sampling rate than high-resolution quantizer. This merit may lead
to a new data acquisition paradigm which reconstructs the sparse
signal from over-sampled but low-resolution measurements.

%different from conventional compressed sensing methods

%opens up an opportunity for
%the hardware complexity grows exponentially as the number of bits

%for example, state-of-art Nyquist ADCs can achieve a sampling rate
%of 550 MHz at 12 bits resolution, whereas 3 GHz with 8 bits
%resolution.

%This has inspired recent interest in studying compressed sensing
%based on quantized measurements. Specifically, in this paper, we
%are interested in examine the extreme case where each measurement
%is quantized into one bit of information.

%(e.g.
%\cite{CandesRomberg06,SunGoyal09,DaiPham09,ZymnisBoyd10,KamilovGoyal11})

%in
%\cite{JacquesLaska11,PlanVershynin11,WimalajeewaVarshney12,LaskaBaraniuk12}

Inspired by practical necessity and potential benefits, compressed
sensing based on low-rate quantized measurements has attracted
considerable attention recently, e.g.
\cite{ZymnisBoyd10,WimalajeewaVarshney12,LaskaBaraniuk12}. In
particular, the problem of recovering a sparse or compressible
signal from one-bit measurements
\begin{align}
\boldsymbol{b}=\text{sign}(\boldsymbol{y})=\text{sign}(\boldsymbol{A}\boldsymbol{x})
\label{eq1}
\end{align}
was firstly introduced by Boufounos and Baraniuk in their work
\cite{BoufounosBaraniuk08}, where ``sign'' denotes an operator
that performs the sign function element-wise on the vector, the
sign function returns $1$ for positive numbers and $-1$ otherwise.
Following \cite{BoufounosBaraniuk08}, the reconstruction
performance from one-bit measurements was more thoroughly studied
and a variety of one-bit compressed sensing algorithms such as
binary iterative hard thresholding (BIHT) \cite{JacquesLaska11},
matching sign pursuit (MSP) \cite{Boufounos09}, and many others
\cite{PlanVershynin11,YanYang12,ShenFang13} were proposed. Despite
all these efforts, very little is known about the optimal choice
of the quantization thresholds and its impact on the recovery
performance. In most previous work (e.g
\cite{BoufounosBaraniuk08,JacquesLaska11,Boufounos09,PlanVershynin11,LaskaWen11,YanYang12,ShenFang13}),
the quantization threshold is set equal to zero. Nevertheless,
such a choice does not necessarily lead to the best reconstruction
accuracy. In fact, in this case, only the sign of the measurement
is retained while the information about the magnitude of the
signal is lost. Hence an exact reconstruction of the sparse signal
is impossible without additional information regarding the sparse
signal. In some other scenarios such as intensity-based source
localization in sensor networks, since all sampled measurements
are non-negative, comparing the original measurements with zero
yields all ones wherever the sources are located, which makes
identifying the true locations impossible. From the above
discussions, we can see that quantization is an integral part of
the sparse signal recovery and is critical to the recovery
performance.

%Similar observations were also reported in

%where researchers made efforts to reduce the reconstruction
%distortion through the quantizer design

There have been some interesting work
\cite{SunGoyal09,KamilovGoyal11,KamilovBourguard12} on quantizer
design for compressed sensing. Specifically, the work
\cite{SunGoyal09} utilizes the high-resolution distributed
functional scalar quantization theory for quantizer design, where
the quantization error is modeled as random noise following a
certain distribution. Nevertheless, such modeling holds valid only
for high-resolution quantization, and may bring limited benefits
when a low-rate quantization strategy is adopted. In
\cite{KamilovGoyal11,KamilovBourguard12}, authors proposed a
generalized approximate message passing (GAMP) algorithm for
quantized compressed sensing, and studied the quantizer design
under the GAMP reconstruction. The method
\cite{KamilovGoyal11,KamilovBourguard12} was developed in a
Bayesian framework by modeling the sparse signal as a random
variable and minimizing the average reconstruction error with
respective to all possible realizations of the sparse signal. This
method, however, involves high-dimensional integral or
multidimensional search to find the optimal quantizer, which makes
a close-form expression of the optimal quantizer design difficult
to obtain. In \cite{KamilovBourguard12}, an adaptive algorithm was
proposed which adjusts the thresholds such that the hyperplanes
pass through the center of the mass of the probability
distribution of the estimated signal. Albeit intuitive, no
rigorous theoretical guarantee was available to justify the
proposed adaptive method. The problem of reconstructing image from
one-bit quantized data is considered in \cite{BourquardAguet10},
where the quantization thresholds are set such that the binary
measurements are acquired with equal probability. Such a scheme
has been shown to be empirically effective. Nevertheless, it is
still unclear why it leads to better reconstruction performance.

%Empirical results confirm that well-chosen thresholds improve the
%reconstruction significantly.
%
%and require alternative approximation methods to circumvent the
%computational difficulty.

In this paper, we investigate the problem of one-bit quantization
design for recovery of sparse signals. The study is conducted in a
deterministic framework by treating the sparse signal as a
deterministic parameter. We provide an quantitative analysis which
examines the choice of the quantization thresholds and quantifies
its impact on the reconstruction performance. Analyses also reveal
that the reconstruction error can be made arbitrarily small by
setting the quantization thresholds close enough to the original
unquantized measurements $\boldsymbol{y}$. Since the original
unquantized data samples $\boldsymbol{y}$ are inaccessible by us,
to address this issue, we propose an adaptive quantization
approach which iteratively adjusts the quantization thresholds
such that the thresholds eventually come close to the original
data samples and thus a small reconstruction error is achieved.

%The magnitude ambiguity issue that arises in
%\cite{BoufounosBaraniuk08,JacquesLaska11,PlanVershynin11,WimalajeewaVarshney12,LaskaBaraniuk12}
%can also be resolved.

%We note that the problem of quantizer design was studied in
%\cite{SunGoyal09,KamilovGoyal11,KamilovBourguard12}

The rest of the paper is organized as follows. In Section
\ref{sec:problem-formulation}, we introduce the one-bit compressed
sensing problem. The main results of this paper are presented in
Section \ref{sec:analysis}, where we show that, by properly
selecting the quantization thresholds, sparse signals can be
recovered from one-bit measurements with an arbitrarily small
error. A rigorous proof of our main result is provided in Section
\ref{sec:proof}. An adaptive quantization method is developed in
\ref{sec:algorithm} to iteratively refine the thresholds based on
previous estimate. In Section \ref{sec:numerical-results},
numerical results are presented to corroborate our theoretical
analysis and illustrate the effectiveness of the proposed
algorithm, followed by concluding remarks in Section
\ref{sec:conclusion}.

%the measurements are quantized into one or only a few number of
%bits.
%Clearly, the canonical form of one-bit compressive sensing can be
%considered as a special case with
%$\boldsymbol{\tau}=\boldsymbol{0}$.

\section{Problem Formulation} \label{sec:problem-formulation}
We consider a coarse quantization-based signal acquisition model
in which each real-valued sample is encoded into one-bit of
information
\begin{align}
\boldsymbol{b}=\text{sign}(\boldsymbol{y}-\boldsymbol{\tau})
=\text{sign}(\boldsymbol{A}\boldsymbol{x}-\boldsymbol{\tau})
\label{eq1}
\end{align}
where
$\boldsymbol{b}=[b_1\phantom{0}b_2\phantom{0}\ldots\phantom{0}b_m]^T$
are the binary observations, and
$\boldsymbol{\tau}=[\tau_1\phantom{0}\tau_2\phantom{0}\ldots\phantom{0}\tau_m]^T$
denotes the quantization threshold vector. As mentioned earlier,
in most previous one-bit compressed sensing studies, the
quantization thresholds are set equal to zero, i.e.
$\boldsymbol{\tau}=\boldsymbol{0}$. A major issue with the zero
quantization threshold is that it introduces a magnitude ambiguity
which cannot be resolved without additional information regarding
the sparse signal. The reason is that when
$\boldsymbol{\tau}=\boldsymbol{0}$, multiplying $\boldsymbol{x}$
by any arbitrary nonzero scaling factor will result in the same
quantized data $\{b_n\}$. To circumvent the magnitude ambiguity
issue, in \cite{JacquesLaska11,PlanVershynin11}, a unit-norm
constraint is imposed on the sparse signal and it was shown that
unit-norm signals can be recovered with a bounded error from one
bit quantized data. Nevertheless, the choice of the zero threshold
is not necessarily the best. In this paper, we are interested in
examining the choice of the quantization thresholds and its impact
on the reconstruction performance.

%fact, for the one-bit compressed sensing, the following
%fundamental questions still remain open: what are the optimal
%quantization thresholds, and how close can the reconstructed
%signal be to the original signal when the best quantization
%thresholds are used?

%The problem of interest in this paper is to study the one-bit
%quantizer design and its impact on the signal recovery
%performance.

%or at least, its optimality has not been rigorously established

%These two questions will be partially addressed in this paper.

Specifically, we consider the following canonical form for sparse
signal reconstruction
\begin{align}
\min_{\boldsymbol{z}}\|\boldsymbol{z}\|_0\quad
\text{s.t.}\phantom{0}\text{sign}(\boldsymbol{A}\boldsymbol{z}-\boldsymbol{\tau})=\boldsymbol{b}\label{opt-1}
\end{align}
Such an optimization problem, albeit non-convex, is more amenable
for theoretical analysis than its convex counterpart
(\ref{opt-2}). Suppose $\boldsymbol{\hat{x}}$ is a solution of the
above optimization. In the following, we examine the choice of
$\boldsymbol{\tau}$ and its impact on the reconstruction error
$\|\boldsymbol{x}-\boldsymbol{\hat{x}}\|_2$.

%only represent the signs of the original samples $\{y_n\}$.

%in which case the quantized data $\{b_n\}$ represent the signs of
%the original samples $\{y_n\}$.

%Now that the only information available is the sign vector after
%the quantization process, we hope that the reconstructed
%$\boldsymbol{\hat{x}}$ will yield outcome consisting with
%knowledge we have after the same sampling and quantization
%process,i.e.
%\begin{align}
%\text{sign}(\boldsymbol{A}\boldsymbol{\hat{x}}-\boldsymbol{\tau})=\boldsymbol{b}
%\end{align}
%Clearly, it is impossible to work out the above problem with a
%unique solution. The restraint that $\boldsymbol{x}$ is sparse
%should also be taken into consideration. In a word, the problem
%can be formulated as
%\begin{align}
%\min_{\boldsymbol{z}}\|\boldsymbol{z}\|_0\quad
%s.t.\phantom{0}\text{sign}(\boldsymbol{A}\boldsymbol{z}-\boldsymbol{\tau})=\boldsymbol{b}\label{eq1}
%\end{align}
%where we adopt the $\ell_1$-norm as sparse-promoting function for
%efficiency of solving the optimization problem.

%In this section, we provide theoretical analysis concerning the
%best choice of threshold for the quantizer, we show that with
%threshold properly set, magnitude information can be preserved and
%hence accurate reconstruction is possible. Besides, we establish
%condition under which the asymptotic reliable recovery of the
%original signal is guaranteed.

%Intuitively, we wish that the quantized binary data have the
%maximum amount of information for signal recovery. We can gain an
%insight into the quantization design from an information-theoretic
%perspective,

\section{One-Bit Quantization Design: Analysis}
\label{sec:analysis}
To facilitate our analysis, we decompose the
quantization threshold vector $\boldsymbol{\tau}$ into a sum of
two terms:
\begin{align}
\boldsymbol{\tau}=\boldsymbol{y}-\boldsymbol{\delta} \label{eq2}
\end{align}
where
$\boldsymbol{\delta}=[\delta_1\phantom{0}\delta_2\phantom{0}\ldots\phantom{0}\delta_m]^T$
can be treated as a constrained deviation from the original
unquantized measurements $\boldsymbol{y}$. Substituting
(\ref{eq2}) into (\ref{eq1}), we get
\begin{align}
\boldsymbol{b}=\text{sign}(\boldsymbol{A}\boldsymbol{x}-\boldsymbol{\tau})=\text{sign}(\boldsymbol{\delta})
\label{eq3}
\end{align}
Suppose $\hat{\boldsymbol{x}}$ is the solution of (\ref{opt-1}).
Let $\boldsymbol{h}\triangleq\hat{\boldsymbol{x}}-\boldsymbol{x}$
be the residual (reconstruction error) vector. Clearly,
$\boldsymbol{h}$ is a $2K$-sparse vector which has at most $2K$
nonzero entries since $\hat{\boldsymbol{x}}$ has at most $K$
nonzero coefficients. Also, the solution $\hat{\boldsymbol{x}}$
yields estimated measurements that are consistent with the
observed binary data, i.e.
\begin{align}
\boldsymbol{b}=\text{sign}(\boldsymbol{A}\boldsymbol{\hat{x}}-\boldsymbol{\tau})=
\text{sign}(\boldsymbol{Ah}+\boldsymbol{\delta}) \label{eq4}
\end{align}
Combining (\ref{eq3}) and (\ref{eq4}), the residual vector
$\boldsymbol{h}$ has to satisfy the following constraint
\begin{align}
\text{sign}(\boldsymbol{Ah}+\boldsymbol{\delta})=\text{sign}(\boldsymbol{\delta})\label{condition-1}
\end{align}
In the following, we will show that the residual vector
$\boldsymbol{h}$ can be bounded by the deviation vector
$\boldsymbol{\delta}$ if the number of measurements are
sufficiently large and $\boldsymbol{A}$ satisfies a certain
condition. Thus the reconstruction error $\boldsymbol{h}$ can be
made arbitrarily small by setting $\boldsymbol{\delta}$ close to
zero. Our main results are summarized as follows.

%with its entries randomly and independently generated from a
%Gaussian distribution

%the Restricted Isometry Property, i.e.

\newtheorem{theorem}{Theorem}
\begin{theorem} \label{theorem1}
Let $\boldsymbol{x}\in\mathbb{R}^n$ be an $K$-sparse vector.
$\boldsymbol{A}\in\mathbb{R}^{m\times n}$ is the sampling matrix.
Suppose there exist an integer $\kappa\geq 2K$ and a positive
parameter $\mu$ such that any $\kappa\times n$ submatrix
$\boldsymbol{\bar{A}}$ constructed by selecting certain rows of
$\boldsymbol{A}$ satisfies
\begin{align}
\|\boldsymbol{\bar{A}}\boldsymbol{u}\|_2^2\geq
\mu\|\boldsymbol{u}\|_2^2 \label{theorem1:condition1}
\end{align}
for all $2K$-sparse vector $\boldsymbol{u}$. Also, assume that
each entry of $\boldsymbol{\delta}$ is independently generated
from a certain distribution with equal probabilities being
positive or negative. Let $\boldsymbol{\hat{x}}$ denote the
solution of the optimization problem (\ref{opt-1}). For any
arbitrarily small value $\eta>0$, we can ensure that the following
statement is true with probability exceeding $1-\eta$: if the
number of measurements, $m$, is sufficiently large and satisfies
\begin{align}
&m-2K\log (m-\kappa+1)-(\kappa-1)\log m \nonumber\\
&\geq  2K (\log (ne^2)-2\log(2K)+1)+\log(1/\eta)+c
\label{theorem1:condition2}
\end{align}
then the sparse signal can be recovered from (\ref{opt-1}) with
the reconstruction error bounded by
\begin{align}
\|\boldsymbol{\hat{x}}-\boldsymbol{x}\|_2\leq
\frac{\epsilon}{\sqrt{\mu}}\triangleq\lambda\epsilon
\label{theorem1:errorbound}
\end{align}
where $\epsilon\triangleq\|\boldsymbol{\delta}\|_2$,
$\lambda\triangleq 1/\sqrt{\mu}$, $e$ in
(\ref{theorem1:condition2}) represents the base of the natural
logarithm, and $c$ in (\ref{theorem1:condition2}) is defined as
\begin{align}
c\triangleq(\kappa-1)(\log(e/(\kappa-1))+1)
\end{align}
which is a constant only dependent on $\kappa$.
\end{theorem}

%The proof of Theorem \ref{theorem1} is provided in the next
%section.

%The adoption of $\mu_{2K}$ instead of $\gamma_{2K}$ render our
%restriction on $\boldsymbol{\bar{A}}\in \mathbb{R}^{\kappa\times
%n}$ less restrictive, with the elimination of the upper bound.

%for a specific submatrix $\boldsymbol{\bar{A}}$,

%Nevertheless, such a condition may be met even for a moderate
%value of $\kappa$. Note that the condition
%(\ref{theorem1:condition1}) only retains the lower bound part of
%the well-known restricted isometry property
%\begin{align}
%(1+\gamma_{2K})\|\boldsymbol{u}\|_2^2\geq\|\boldsymbol{\bar{A}u}\|_2^2\geq(1-\gamma_{2K})\|\boldsymbol{u}\|_2^2
%\end{align}
%which holds for all $2K$-sparse vector $\boldsymbol{u}$. Thus we
%have $\mu_{2K}\leq\gamma_{2K}$.

We have the following remarks regarding Theorem \ref{theorem1}.

\emph{Remark 1:} Clearly, the condition
(\ref{theorem1:condition1}) is guaranteed if there exists a
constant $\gamma_{2K}<1$ such that any $\kappa\times n$ submatrix
of $\boldsymbol{A}$ satisfies
\begin{align}
(1+\gamma_{2K})\|\boldsymbol{u}\|_2^2\geq\|\boldsymbol{\bar{A}u}\|_2^2\geq(1-\gamma_{2K})\|\boldsymbol{u}\|_2^2
\end{align}
for all $2K$-sparse vector $\boldsymbol{u}$. As indicated in
previous studies \cite{CandesTao05}, when the sampling matrix is
randomly generated from a Gaussian or Bernoulli distribution, this
restricted isometry property is met with an overwhelming
probability, particularly when $\kappa$ is large. Hence it is not
difficult to find an integer $\kappa$ and a positive $\mu$ such
that the condition (\ref{theorem1:condition1}) holds.

%Different submatrices may have different restricted isometry
%constants. Nevertheless, we can simply let $\mu$ equal to the
%largest restricted isometry constant, which ensures the condition
%(\ref{theorem1:condition1}) holds valid for all possible
%submatrices.

%For different construction of the submatrix
%$\boldsymbol{\bar{A}}$, the corresponding may vary from case to
%case, here we can simply apply the largest $\mu_{2K}$ for all the
%possible $\boldsymbol{\bar{A}}$, in this case, the condition in
%(\ref{theorem1:condition1}) will hold for all the possible
%construction of $\boldsymbol{\bar{A}}$. Therefore, as long as we
%can find out a $\kappa$ such that there exists some constant
%$\mu_{2K}<1$ satisfying (\ref{theorem1:condition1}), the matrix
%$\boldsymbol{A}$ could be a suitable sampling matrix.

%A sufficient condition is provided in (\ref{theorem1:condition2})
%which guarantees exact reconstruction of the sparse signal via
%solving the optimization problem in (\ref{opt-1}).

%can be lower bounded by
%\begin{align}
%&m-2K\log (m-\kappa+1)-(\kappa-1)\log m\nonumber\\
%&\geq m- 2K\log m -(\kappa-1)\log m
%\end{align}
%such lower bound presented above

\emph{Remark 2:}  Note that the term on the left hand side of
(\ref{theorem1:condition2}) is a monotonically increasing function
of $m$ when $m$ is greater than a certain value. Thus for fixed
values of $K$, $\kappa$, $n$, and $\eta$, we can always find a
sufficiently large $m$ to ensure the condition
(\ref{theorem1:condition2}) is guaranteed. In addition, a close
inspection of (\ref{theorem1:condition2}) reveals that, for a
constant $\eta$, the number of measurements required to guarantee
(\ref{theorem1:condition2}) is of order $\mathcal{O}(K\log(n))$,
which is the same as that for conventional compressed sensing.

%and in turn, the residual vector $\boldsymbol{h}$ is bounded.

\emph{Remark 3:} From (\ref{theorem1:errorbound}), we see that the
sparse signal can be recovered with an arbitrarily small error by
letting $\boldsymbol{\delta}\rightarrow\boldsymbol{0}$, or in
other words, by setting the threshold vector $\boldsymbol{\tau}$
sufficiently close to the unquantized data $\boldsymbol{y}$. This
result has two important implications. Firstly, sparse signals can
be reconstructed with negligible errors even from one-bit
measurements. We note that \cite{JacquesLaska11,PlanVershynin11}
also showed that sparse signals can be recovered with a bounded
error from one bit quantized data. But their analysis ignores the
magnitude ambiguity and confines sparse signals on the unit
Euclidean sphere. Such a sphere constraint is no longer necessary
for our analysis. Secondly, our result suggests that the best
thresholds should be set as close to the original data samples
$\boldsymbol{y}$ as possible. So far there has no theoretical
guarantee indicating that other choice of thresholds can also lead
to a stable recovery with an arbitrarily small reconstruction
error.

%In contrast, our analysis shows that the magnitude ambiguity can
%be removed by choosing nonzero thresholds.
%To the best of our knowledge, this is the first time such a claim
%is made.
%have to be normalized in order to remove the magnitude ambiguity
%arising from the sign operation

%and a unit-norm constraint has to be imposed on the sparse signal
%does not take into account

%reveals that given a sufficient large number of nonlinear samples
%of a sparse signal, the original signal can be recovered by
%(\ref{opt-1}), and the upper bound of reconstruction error is
%proportional to $\epsilon$, the norm of the deviation between the
%threshold and the real valued measurements. Consequently, as long
%as we can set the threshold $\boldsymbol{\tau}$ as close to the
%original measurements as possible, the sparse signal can be
%reliably recovered with error approaching to zero.

%where $\delta_{2K}$ is the restricted isometry constant associated
%with the measurement matrix $\boldsymbol{A}$, which is defined as
%the smallest constant such that
%\begin{align}
%(1-\delta_{2K})\|\boldsymbol{z}\|_2^2\leq\|\boldsymbol{Az}\|_2^2\leq(1+\delta_{2K})\|\boldsymbol{z}\|_2^2\label{RIP}
%\end{align}
%holds for all $2K$-sparse vectors \cite{CandesTao05}.

\section{Proof of Theorem \ref{theorem1}} \label{sec:proof}
Let $\boldsymbol{a}_i^T$ denote the $i$th row of the sampling
matrix $\boldsymbol{A}$. Clearly, to ensure the sign consistency
(\ref{condition-1}) for each component, we should either have
\begin{align}
\text{sign}(\boldsymbol{a}_i^T\boldsymbol{h})=\text{sign}(\delta_i)
\label{condition-2}
\end{align}
or
\begin{align}
|\boldsymbol{a}_i^T\boldsymbol{h}|<|\delta_i| \label{condition-3}
\end{align}
Note that in the latter case the sign of
$(\boldsymbol{a}_i^T\boldsymbol{h}+\delta_i)$ does not flip
regardless of the sign of $\boldsymbol{a}_i^T\boldsymbol{h}$.
Therefore if there exists a $2K$-sparse residual vector
$\boldsymbol{h}$ such that each component of $\boldsymbol{Ah}$
satisfies either (\ref{condition-2}) or (\ref{condition-3}), then
$\boldsymbol{\hat{x}}=\boldsymbol{x}+\boldsymbol{h}$ is the
solution of (\ref{opt-1}). Our objective in the following is to
show that such a residual vector $\boldsymbol{h}$ is bounded by
the deviation vector $\boldsymbol{\delta}$ with an arbitrarily
high probability.

%where $c$ is a parameter associated with the sampling matrix
%$\boldsymbol{A}$.

%Let $T$ denote the support of $\boldsymbol{h}$, $\boldsymbol{h}_T$
%denote a $2K$-dimensional vector constructed by stacking the
%nonzero components $\{h_i\}$, $\forall \{i|h_i\neq 0\}$.
%Similarly, $\boldsymbol{A}_T\in \mathbb{R}^{m\times 2K}$
%represents a submatrix of $\boldsymbol{A}$ obtained by
%concatenating the columns associated with the index set $T$. We
%have
%\begin{align}
%\boldsymbol{Ah}=\boldsymbol{A}_T\boldsymbol{h}_T
%\end{align}

Without loss of generality, we decompose $\boldsymbol{A}$ and
$\boldsymbol{\delta}$ into two parts:
$\boldsymbol{A}=[\boldsymbol{A}_{1}^T\phantom{0}\boldsymbol{A}_{2}^T]^T$,
$\boldsymbol{\delta}=[\boldsymbol{\delta}_{1}^T\phantom{0}\boldsymbol{\delta}_{2}^T
]^T$, according to the two possible relationships between
$\boldsymbol{a}_i^T\boldsymbol{h}$ and $\delta_i$:
\begin{eqnarray}
\left.
\begin{array}{cc}
\text{sign}(\boldsymbol{A}_{1}\boldsymbol{h})=\text{sign}(\boldsymbol{\delta}_{1})
\\
|\boldsymbol{A}_{2}\boldsymbol{h}|< |\boldsymbol{\delta}_{2}|
 \end{array} \right. \label{event}
\end{eqnarray}
where $\boldsymbol{A}_{1}\in\mathbb{R}^{m_1\times n}$,
$\boldsymbol{A}_{2}\in\mathbb{R}^{m_2\times n}$, $m_1+m_2=m$, and
in the second equation, both the absolute value operation
$|\cdot|$ and the inequality symbol ``$<$'' applies entrywise to
vectors.

%``E1'' and ``E2'' represent event 1 and event 2, respectively.

%(|\boldsymbol{A}_{T_2}\boldsymbol{h}|<
%|\boldsymbol{\delta}_{T_2}|) \cap
%(|\boldsymbol{A}_{T_2}\boldsymbol{h}|<
%|\boldsymbol{\delta}_{T_2}|) \cap

We now analyze the probability of the residual vector
$\boldsymbol{h}$ being greater than $\lambda\epsilon$, given that
the condition (\ref{event}) is satisfied. This conditional
probability can be denoted as
$P(\|\boldsymbol{h}\|_2>\lambda\epsilon|\text{E})$, where we use
``E'' to denote the event (\ref{event}). Clearly,
$P(\|\boldsymbol{h}\|_2>\lambda\epsilon|\text{E})$ can also be
explained as the probability that the event
$\|\boldsymbol{h}\|_2>\lambda\epsilon$ will occur, when the event
$\text{E}$ has occurred. To facilitate our analysis, we divide the
event (\ref{event}) into two disjoint sub-events which are defined
as
\begin{align}
&\text{E1: The event (\ref{event}) holds true for $m_2\geq
\kappa$} \nonumber
\\
&\text{E2: The event (\ref{event}) holds true for $m_2<\kappa$}
\nonumber
\end{align}
Clearly, the union of these two sub-events is equal to the event
(\ref{event}). Utilizing Bayes' Theorem, the probability
$P(\|\boldsymbol{h}\|_2>\lambda\epsilon|\text{E})$ can be
expressed as
\begin{align}
P(\|\boldsymbol{h}\|_2>\lambda\epsilon|\text{E})
\stackrel{(a)}{=}&P(\|\boldsymbol{h}\|_2>\lambda\epsilon,\text{E})
\nonumber\\
=&P(\|\boldsymbol{h}\|_2>\lambda\epsilon,\text{E1}\cup\text{E2})
\nonumber\\
\stackrel{(b)}{=}&P(\|\boldsymbol{h}\|_2>\lambda\epsilon,\text{E1})+P(\|\boldsymbol{h}\|_2>\lambda\epsilon,\text{E2})
%\nonumber\\
%=&P(\|\boldsymbol{h}\|_2>\lambda\epsilon|\text{Ea})P(\text{Ea})\nonumber\\
%&+ P(\|\boldsymbol{h}\|_2>\lambda\epsilon|\text{Eb})P(\text{Eb})
\label{eq8}
\end{align}
where $(a)$ holds because the events (\ref{event}) is a
prerequisite condition that is always met to ensure the sign
consistency, i.e. $P(\text{E})=1$, and $(b)$ follows from the fact
that the probability of the union of two disjoint events is equal
to the sum of their respective probabilities.

%which can be expressed as
%\begin{align}
%P(\|\boldsymbol{h}\|_2>\lambda\epsilon|\text{Ea})
%=\frac{P(\|\boldsymbol{h}\|_2>\lambda\epsilon,\text{E1},\text{E2a})}{P(\text{E1},\text{E2a})}
%\end{align}

%are full rank with the singular values greater than $\eta$
%=&\|\boldsymbol{\tilde{A}}_{2}\boldsymbol{h}_T\|^2=\boldsymbol{h}_T^T\boldsymbol{\tilde{A}}_{2}^T\boldsymbol{\tilde{A}}_{2}\boldsymbol{h}_T
%\nonumber\\
%=&\boldsymbol{h}_T^T(\boldsymbol{\tilde{A}}_{2,1}^T\boldsymbol{\tilde{A}}_{2,1}+
%\boldsymbol{\tilde{A}}_{2,2}^T\boldsymbol{\tilde{A}}_{2,2})\boldsymbol{h}_T
%\nonumber\\
%\geq&\boldsymbol{h}_T^T\boldsymbol{\tilde{A}}_{2,1}^T\boldsymbol{\tilde{A}}_{2,1}\boldsymbol{h}_T

%where $\boldsymbol{h}_T\in\mathbb{R}^{2K}$ denotes the vector
%formed by stacking the nonzero components of $\boldsymbol{h}$,
%$\boldsymbol{\tilde{A}}_2\in\mathbb{R}^{m_2\times 2K}$ represents
%a submatrix of $\boldsymbol{A}_2$ obtained by concatenating the
%columns associated with the nonzero entries in $\boldsymbol{h}$,
%and
%$\boldsymbol{\tilde{A}}_{2}=[\boldsymbol{\tilde{A}}_{2,1}^T\phantom{0}\boldsymbol{\tilde{A}}_{2,2}^T]^T$,
%where $\boldsymbol{\tilde{A}}_{2,1}\in\mathbb{R}^{2K\times 2K}$,
%and $\boldsymbol{\tilde{A}}_{2,1}\in\mathbb{R}^{(m_2-2K)\times
%2K}$.

Let us first examine the probability
$P(\|\boldsymbol{h}\|_2>\lambda\epsilon,\text{E1})$. We show that
the events $\|\boldsymbol{h}\|_2>\lambda\epsilon$ and E1 are two
mutually exclusive events which cannot occur at the same time. To
see this, note that when the event E1 occurs, we should have
\begin{align}
|\boldsymbol{A}_{2}\boldsymbol{h}|<
|\boldsymbol{\delta}_{2}|\Rightarrow\|\boldsymbol{A}_{2}\boldsymbol{h}\|_2^2<\epsilon^2
\label{eq5}
\end{align}
in which $\boldsymbol{A}_{2}\in\mathbb{R}^{m_2\times n}$ and
$m_2\geq \kappa$. On the other hand, recalling that any
$\kappa\times n$ sub-matrix formed by selecting certain rows of
$\boldsymbol{A}$ satisfies the condition
(\ref{theorem1:condition1}), we have
\begin{align}
\|\boldsymbol{A}_{2}\boldsymbol{h}\|_2^2
\geq\mu\|\boldsymbol{h}\|_2^2 \label{eq6}
\end{align}
Combining (\ref{eq5}) and (\ref{eq6}), we arrive at
\begin{align}
\|\boldsymbol{h}\|_2<\frac{\epsilon}{\sqrt{\mu}}\triangleq\lambda\epsilon
\end{align}
which is contradictory to the event
$\|\boldsymbol{h}\|_2>\lambda\epsilon$. Hence the events E1 and
$\|\boldsymbol{h}\|_2>\lambda\epsilon$ cannot occur
simultaneously, which implies
\begin{align}
P(\|\boldsymbol{h}\|_2>\lambda\epsilon,\text{E1})=0 \label{eq7}
\end{align}

%As a result, the probability of $\boldsymbol{h}>\lambda\epsilon$
%given E1 and E2a is equal to zero, i.e.
%\begin{align}
%P(\|\boldsymbol{h}\|_2>\lambda\epsilon|\text{E1},\text{E2a}) =0
%\label{eq7}
%\end{align}

%which is the intersection of the event (\ref{event}) and the event
%$m_2<2K$, and the event (\ref{event}) is an intersection of two
%sub-events defined in (\ref{event})

%may or may not have the same signs as the corresponding entries in
%$\boldsymbol{\delta}$

%According to the possible combinations of the $m-2K$ components
%which are meant to meet the sign consistency requirement

Substituting (\ref{eq7}) into (\ref{eq8}), the probability
$P(\|\boldsymbol{h}\|_2>\lambda\epsilon|\text{E})$ is simplified
as
\begin{align}
P(\|\boldsymbol{h}\|_2>\lambda\epsilon|\text{E}) =
P(\|\boldsymbol{h}\|_2>\lambda\epsilon,\text{E2}) \label{eq20}
\end{align}
The probability
$P(\|\boldsymbol{h}\|_2>\lambda\epsilon,\text{E2})$, however, is
still difficult to analyze. To circumvent this difficulty, we,
instead, derive an upper bound on
$P(\|\boldsymbol{h}\|_2>\lambda\epsilon,\text{E2})$:
\begin{align}
&P(\|\boldsymbol{h}\|_2>\lambda\epsilon,\text{E2}) \nonumber\\
=&P(\|\boldsymbol{h}\|_2>\lambda\epsilon|\text{E2}) P(\text{E2})
\leq P(\text{E2}) \nonumber\\
\stackrel{(a)}{=}&
P(\text{sign}(\boldsymbol{A}_{1}\boldsymbol{h})=\text{sign}(\boldsymbol{\delta}_{1}),
|\boldsymbol{A}_{2}\boldsymbol{h}|< |\boldsymbol{\delta}_{2}|, m_2<\kappa) \nonumber\\
\leq&
P(\text{sign}(\boldsymbol{A}_{1}\boldsymbol{h})=\text{sign}(\boldsymbol{\delta}_{1}),
m_2<\kappa) \nonumber\\
=&
P(\text{sign}(\boldsymbol{A}_{1}\boldsymbol{h})=\text{sign}(\boldsymbol{\delta}_{1}),
m_1> m-\kappa) \nonumber\\
\stackrel{(b)}{=}&P\bigg(\bigcup_{i=0}^{\kappa-1}\Omega_i\bigg)
\label{eq10}
\end{align}
where $(a)$ comes from the definition of the event E2, and in
$(b)$, the event $\Omega_i$ is defined as
\begin{align}
\Omega_i:\phantom{0}
&\text{sign}(\boldsymbol{A}_{1}\boldsymbol{h})=\text{sign}(\boldsymbol{\delta}_{1})
\nonumber\\
&\text{where $\boldsymbol{A}_{1}\in\mathbb{R}^{m_1\times n}$, and
$m_1=m-i$}
\end{align}
which means that there exist at least $m-i$ components in
$\boldsymbol{Ah}$ whose signs are consistent with the
corresponding entries in $\boldsymbol{\delta}$. Note that since
the rest $i$ components are not explicitly specified in
$\Omega_i$, the event $\Omega_i$ include all possibilities for the
rest $i$ components. Based on the definition of $\Omega_i$, we can
infer the following relationship:
$\Omega_{i_1}\supset\Omega_{i_2}$ for $i_1>i_2$. This is because
$\Omega_{i_2}$ can be regarded as a special case of the event
$\Omega_{i_1}$ with some of the unspecified $i_1$ components also
meeting the sign consistency requirement. With this relation, the
upper bound derived in (\ref{eq10}) can be simplified as
\begin{align}
P(\|\boldsymbol{h}\|_2>\lambda\epsilon,\text{E2})\leq
P\bigg(\bigcup_{i=0}^{\kappa-1}\Omega_i\bigg)
=P(\Omega_{\kappa-1}) \label{eq19}
\end{align}
We note that for the event $\Omega_{\kappa-1}$,
$\boldsymbol{A}_{1}$ is not specified and can be any submatrix of
$\boldsymbol{A}$. Considering selection of $m-\kappa+1$ rows (out
of $m$ rows of $\boldsymbol{A}$) to construct
$\boldsymbol{A}_{1}$, the event $\Omega_{\kappa-1}$ can be
expressed as a union of a set of sub-events
\begin{align}
\Omega_{\kappa-1}=\bigcup_{j=1}^J \Omega_{\kappa-1}^{j}
\label{eq11}
\end{align}
where $J\triangleq C(m,m-\kappa+1)$, $C(m,k)$ denotes the number
of $k$ combinations from a given set of $m$-elements, and each
sub-event $\Omega_{\kappa-1}^{j}$ is defined as
\begin{align}
\Omega_{\kappa-1}^{j}:
\text{sign}(\boldsymbol{A}_{1}\boldsymbol{h})=\text{sign}(\boldsymbol{\delta}_{1})\phantom{0}
\text{where $\boldsymbol{A}_{1}=\boldsymbol{A}[I_j,:]$}
\end{align}
in which $I_j$ is an unique index set which consists of
$m-\kappa+1$ non-identical indices selected from
$\{1,2,\ldots,m\}$, $\boldsymbol{A}[I_j,:]$ denotes a submatrix of
$\boldsymbol{A}$ constructed by certain rows from
$\boldsymbol{A}$, and the indices of the selected rows are
specified by $I_j$. From (\ref{eq11}), we have
\begin{align}
P(\Omega_{\kappa-1})=&P\bigg(\bigcup_{j=1}^J
\Omega_{\kappa-1}^{j}\bigg) \stackrel{(a)}{\leq} \sum_{j=1}^J
P(\Omega_{\kappa-1}^{j}) \label{eq12}
\end{align}
where the inequality $(a)$ follows from the fact that the
probability of a union of events is no greater than the sum of
probabilities of respective events. The inequality becomes an
equality if the events are disjoint. Nevertheless, the sub-events
$\{\Omega_{\kappa-1}^{j}\}$ are not necessarily disjoint and may
have overlappings due to the $\kappa-1$ unspecified components.

%Note that the event $\text{e}_{k}$ requires sign consistency for
%only $(m-k)$ components in $\boldsymbol{Ah}$, leaving the rest
%$k_1$ components unconstrained.
%it can be expressed in another way
%\begin{align}
%\text{e}_k: \quad &\text{At least $(m-k)$ components in
%$\boldsymbol{Ah}$ have their signs } \nonumber \\
%& \text{consistent with those in $\boldsymbol{\delta}$} \nonumber
%\end{align}
%Based on this definition,

%can be expressed as a union of two events:
%\begin{align}
%\text{e}_{k_1}=\text{e}_{k_1}^1\cup\text{e}_{k_1}^2
%\end{align}
%where $\text{e}_{k_1}^1$ and $\text{e}_{k_1}^2$ are respectively
%defined as
%\begin{align}
%\text{e}_{k_1}^1: \quad \text{}
%\end{align}

%Although $P(\|\boldsymbol{h}\|_2>\lambda\epsilon,\text{Eb})$
%involves computing the probability of the union of a set of events
%$\{\text{e}_{k}\}$, a close study of the relationship between
%these events

%For notational convenience, we let Ec denote the intersection of
%the two events:
%$\text{sign}(\boldsymbol{A}_{1}\boldsymbol{h})=\text{sign}(\boldsymbol{\delta}_{1})$
%and $m_2<2K$, i.e.
%\begin{align}
%\text{Ec:}\quad &\text{The events
%$\text{sign}(\boldsymbol{A}_{1}\boldsymbol{h})=\text{sign}(\boldsymbol{\delta}_{1})$
%and $m_2<2K$} \nonumber\\
%&\text{occur simultaneously} \nonumber
%\end{align}

We now analyze the probability $P(\Omega_{\kappa-1}^{j})$. To
begin with our analysis, we introduce the concept of orthant
originally proposed in \cite{JacquesLaska11} for analysis of
one-bit compressed sensing. An orthant in $\mathbb{R}^m$ is a set
of vectors that share the same sign pattern, i.e.
\begin{align}
\mathcal{O}_{\tilde{\boldsymbol{u}}}=\{\boldsymbol{u}\in\mathbb{R}^m|\text{sign}(\boldsymbol{u})=\tilde{\boldsymbol{u}}\}
\end{align}
A useful result concerning intersections of orthants by subspaces
is summarized as follows.
\newtheorem{lemma}{Lemma}
\begin{lemma} \label{lemma1}
Let $\mathcal{S}$ be an arbitrary $k$-dimensional subspace in an
$m$-dimensional space. Then the number of orthants intersected by
$\mathcal{S}$ can be upper bounded by
\begin{align}
I(m,k)\leq 2^{k} C(m,k)
\end{align}
where $C(m,k)$ denotes the number of $k$-combinations from a set
of $n$-elements.
\end{lemma}
\begin{proof}
See \cite[Lemma 8]{JacquesLaska11a}.
\end{proof}

%\left({m\atop k}\right)

The probability $P(\Omega_{\kappa-1}^{j})$ of our interest can be
interpreted as, the probability of the vector
$\boldsymbol{A}_{1}\boldsymbol{h}$ lying in the same orthant as
$\boldsymbol{\delta}_{1}$ for a given $\boldsymbol{A}_1$. We first
examine the number of sign patterns the vector
$\boldsymbol{A}_{1}\boldsymbol{h}$ could possibly have. Let $S$
denote the set of all possible sign patterns for
$\boldsymbol{A}_{1}\boldsymbol{h}$, i.e.
\begin{align}
S=\{\boldsymbol{u}=\text{sign}(\boldsymbol{A}_{1}\boldsymbol{h})|&\text{$\boldsymbol{h}\in
\mathbb{R}^{n}$ is a $2K$-sparse vector}, \nonumber\\
& \boldsymbol{A}_{1}=\boldsymbol{A}[I_j,:]\}
\end{align}
Also, let $T$ denote the support of the sparse residual vector
$\boldsymbol{h}$, we can write
\begin{align}
\boldsymbol{g}\triangleq\boldsymbol{A}_{1}\boldsymbol{h}=\boldsymbol{A}_{1}[:,T]\boldsymbol{h}_T
\end{align}
where $\boldsymbol{g}\in\mathbb{R}^{m-\kappa+1}$, and
$\boldsymbol{A}_1[:,T]$ denotes a submatrix of $\boldsymbol{A}_1$
obtained by concatenating columns whose indices are specified by
$T$. We see that $\boldsymbol{g}$ is a linear combination of $2K$
columns of $\boldsymbol{A}_{1}$, and thus $\boldsymbol{g}$ lies in
an $2K$-dimensional subspace spanned by the columns of
$\boldsymbol{A}_{1}[:,T]$. Recalling Lemma \ref{lemma1}, we know
that the number of orthants intersected by this subspace is upper
bounded by $2^{2K}C(m-\kappa+1,2K)$. Therefore the vector
$\boldsymbol{g}$ which lies in this subspace has at most
$2^{2K}C(m-\kappa+1,2K)$ possible sign patterns. Note that this
result is for a specific choice of the index set $T$. The
selection of the support $T$ from $n$ entries has at most
$C(n,2K)$ combinations. Therefore, in summary, the number of sign
patterns in the set $S$ is upper bounded by
\begin{align}
N_{\text{SP}}\leq 2^{2K}C(m-\kappa+1,2K)C(n,2K) \label{eq9}
\end{align}
The probability $P(\Omega_{\kappa-1}^{j})$ can be calculated as
\begin{align}
P(\Omega_{\kappa-1}^{j})=&
P(\text{sign}(\boldsymbol{\delta}_{1})\in S) \nonumber\\
\stackrel{(a)}{=}&\frac{N_{\text{SP}}}{2^{(m-\kappa+1)}}\leq\frac{2^{2K}C(m-\kappa+1,2K)C(n,2K)}{2^{(m-\kappa+1)}}
\label{eq13}
\end{align}
where $(a)$ comes from the fact that $\boldsymbol{\delta}$ is a
vector whose entries are independently generated from a certain
distribution with equal probabilities being positive and negative,
and $\boldsymbol{\delta}_{1}$ has $2^{(m-\kappa+1)}$ possible sign
patterns.

Combining (\ref{eq20}), (\ref{eq19}), (\ref{eq12}) and
(\ref{eq13}), we arrive at
\begin{align}
&P(\|\boldsymbol{h}\|_2>\lambda\epsilon|\text{E})=P(\|\boldsymbol{h}\|_2>\lambda\epsilon|\text{E2})
\leq P(\Omega_{\kappa-1}) \nonumber\\
\leq&
\frac{2^{2K}C(m-\kappa+1,2K)C(n,2K)C(m,\kappa-1)}{2^{(m-\kappa+1)}}
\end{align}
where the last inequality comes from
$J=C(m,m-\kappa+1)=C(m,\kappa-1)$. Utilizing the following
inequality \cite{JacquesLaska11}
\begin{align}
C(a,b)\leq \left(\frac{ae}{b}\right)^b
\end{align}
in which $e\approx2.718$ denotes the base of the natural
logarithm, the probability
$P(\|\boldsymbol{h}\|_2>\lambda\epsilon|\text{E})$ can be further
bounded by
\begin{align}
P(\|\boldsymbol{h}\|_2>\lambda\epsilon|\text{E})\leq\frac{a}{b}
\label{eq15}
\end{align}
where
\begin{align}
a\triangleq\left(\frac{n(m-\kappa+1)e^2}{(2K)^2}\right)^{2K}
\left(\frac{me}{\kappa-1}\right)^{(\kappa-1)} \nonumber
\end{align}
\begin{align}
b\triangleq 2^{(m-\kappa+1-2K)} \nonumber
\end{align}
Examine the condition which ensures that $a/b$ is less than a
specified value $\eta$, where $0<\eta<1$. Taking the base-2
logarithm on both sides of $(a/b)\leq \eta$ and rearranging the
equation, we obtain
\begin{align}
&m-2K\log (m-\kappa+1)-(\kappa-1)\log m \nonumber\\
&\geq  2K (\log (ne^2)-2\log(2K)+1)+\log(1/\eta)+c \label{eq16}
\end{align}
where
\begin{align}
c\triangleq(\kappa-1)(\log(e/(\kappa-1))+1) \nonumber
\end{align}
is a constant only dependent on $\kappa$. In summary, for a
specified $\eta$, if the condition (\ref{eq16}) is satisfied, then
we can ensure that the probability of the residual vector
$\boldsymbol{h}$ being greater than $\lambda\epsilon$ is smaller
than $\eta$, i.e.
\begin{align}
P(\|\boldsymbol{h}\|_2>\lambda\epsilon|\text{E})\leq \eta
\end{align}
or
\begin{align}
P(\|\boldsymbol{h}\|_2\leq\lambda\epsilon|\text{E})\geq 1-\eta
\end{align}
The proof is completed here.

\section{Quantization Design: Adaptive Methods}
\label{sec:algorithm} In this section, we aim to develop a
practical algorithm for one-bit compressed sensing. Previous
analyses show that a reliable and accurate recovery of sparse
signals is possible even based on one-bit measurements. This
theoretical result is encouraging but we are still confronted with
two practical difficulties while trying to recover the sparse
signal via solving (\ref{opt-1}). Firstly, the optimization
(\ref{opt-1}) is a non-convex and NP hard problem that has
computational complexity growing exponentially with the signal
dimension $n$. Hence alternative optimization strategies which are
more computationally efficient in finding the sparse solution are
desirable. Secondly, our theoretical analysis suggests that the
quantization thresholds should be set as close as possible to the
original unquantized measurements $\boldsymbol{y}$. However, in
practice, the decoder does not have access to the unquantzed data
$\boldsymbol{y}$. To overcome this difficulty, we will consider a
data-dependent adaptive quantization approach whereby the
quantization threshold vector is dynamically adjusted from one
iteration to next, in a way such that the threshold come close to
the desired values.

%On the other hand, it has be shown in the analysis above, that the
%reconstruction error is constrained by the perturbation vector
%$\boldsymbol{\delta}$. In this case, the default threshold
%$\boldsymbol{0}$ applied in the conventional $1$-bit compressive
%sensing scheme will certainly fail to meet the requirement.
%However, in practice, we do not have access to the value of the
%original sampling data after quantization, so it is of great
%significance to seek for practical strategies to make a better
%choice of the threshold adopted. We will try to design an
%efficient and reliable algorithm that cope with these two problems
%in this section.

%\subsection{Alternative sparse-encouraging penalty}

\subsection{Computational Issue}
To circumvent the computational issue of (\ref{opt-1}), we can
replace the $\ell_0$-norm with alternative sparsity-promoting
functionals. The most popular alternative is the $\ell_1$-norm.
Replacing the $\ell_0$-norm with this sparsity-encouraging
functional leads to the following optimization
\begin{align}
\min_{\boldsymbol{z}}\|\boldsymbol{z}\|_1\quad
\text{s.t.}\phantom{0}\text{sign}(\boldsymbol{A}\boldsymbol{z}-\boldsymbol{\tau})=\boldsymbol{b}\label{opt-2}
\end{align}
which is convex and can be recast as a linear programming problem
that can be solved efficiently. Although a rigorous theoretical
justification for $\ell_1$-minimization based optimization is
still unavailable\footnote{A theoretical guarantee for
$\ell_1$-minimization is under study and will be provided in our
future work.}, our simulation results indeed suggest that
(\ref{opt-2}) is an effective alternative to $\ell_0$-minimization
and is able to yield a reliable and accurate reconstruction of
sparse signals.

In addition to the $\ell_1$-norm, another alternative
sparse-promoting functional is the log-sum penalty function. The
optimization based on the log-sum penalty function can be
formulated as
\begin{align}
\min_{\boldsymbol{z}}\sum_{i=1}^n \log (|z_i|+\epsilon)\quad
\text{s.t.}\phantom{0}\text{sign}(\boldsymbol{A}\boldsymbol{z}-\boldsymbol{\tau})=\boldsymbol{b}\label{opt-3}
\end{align}
where $\epsilon>0$ is a parameter ensuring that the penalty
function is well-defined. Log-sum penalty function was originally
introduced in \cite{CoifmanWickerhauser92} for basis selection and
has gained increasing attention recently
\cite{CandesWakin08,WipfNagaranjan10}. Experiments and theoretical
analyses show that log-sum penalty function behaves more like
$\ell_0$-norm than $\ell_1$-norm, and has the potential to present
superiority over the $\ell_1$-minimization based methods. The
optimization (\ref{opt-3}) can be efficiently solved by resorting
to a bound optimization technique
\cite{LangeHunter00,CandesWakin08,WipfNagaranjan10}. The basic
idea is to construct a surrogate function
$Q(\boldsymbol{z}|\boldsymbol{\hat{z}}^{(t)})$ such that
\begin{align}
Q(\boldsymbol{z}|\boldsymbol{\hat{z}}^{(t)})-L(\boldsymbol{z})\geq
0
\end{align}
where $L(\boldsymbol{z})$ is the objective function and the
minimum is attained when
$\boldsymbol{z}=\boldsymbol{\hat{z}}^{(t)}$, i.e.
$Q(\boldsymbol{\hat{z}}^{(t)}|\boldsymbol{\hat{z}}^{(t)})=L(\boldsymbol{\hat{z}}^{(t)})$.
Optimizing $L(\boldsymbol{z})$ can be replaced by minimizing the
surrogate function $Q(\boldsymbol{z}|\boldsymbol{\hat{z}}^{(t)})$
iteratively. An appropriate choice of such a surrogate function
for the objective function (\ref{opt-3}) is given by
\cite{CandesWakin08}
\begin{align}
Q(\boldsymbol{z}|\boldsymbol{\hat{z}}^{(t)})=\sum_{i=1}^n\left\{\frac{|z_i|}{|z_i^{(t)}|+\epsilon}+
\log
(|z_i^{(t)}|+\epsilon)-\frac{|z_i^{(t)}|}{|z_i^{(t)}|+\epsilon}\right\}
\end{align}
Therefore optimizing (\ref{opt-3}) can be formulated as reweighted
$\ell_1$-minimization which iteratively minimizes the following
weighted $\ell_1$ function:
\begin{align}
\min_{\boldsymbol{z}}\quad&
Q(\boldsymbol{z}|\boldsymbol{\hat{z}}^{(t)})=\sum_{i=1}^n
w_i^{(t)}|z_i|+\text{constant}\nonumber\\
\text{s.t.}\quad&
\phantom{0}\text{sign}(\boldsymbol{A}\boldsymbol{z}-\boldsymbol{\tau})=\boldsymbol{b}
\label{opt-4}
\end{align}
where the weighting parameters are given by
$w_i^{(t)}=1/(|z_i^{(t-1)}|+\epsilon),\forall i$. The optimization
(\ref{opt-4}) is a weighted version of (\ref{opt-2}) and can also
be recast as a linear programming problem. By iteratively
minimizing (\ref{opt-2}), we can guarantee that the objective
function value $L(\boldsymbol{x})$ is non-increasing at each
iteration. In this manner, the reweighted iterative algorithm
eventually converges to a local minimum of (\ref{opt-3}).

%It has been shown that based on certain bounded optimization
%techniques, minimizing the log-sum term results in iterative
%reweighted $\ell_1$ minimization process, that estimates the
%sparse signal and update the weights of the bounded surrogate
%function alternatively. Moreover, it has been collaborated in a
%lot of numerical results that such an iterative algorithm
%outperforms the popular $\ell_1$-norm minimization. Hence, we may
%expect the replacement of the objective function in (\ref{opt-3})
%would result in better performance than that in (\ref{opt-3}).

%\begin{figure}[h]\centering
%\includegraphics[width=9cm]{fig1.eps}
%\caption{Framework of our proposed algorithm.} \label{fig1}
%\end{figure}

\begin{figure}[!t]
\centering
\includegraphics[width=9cm]{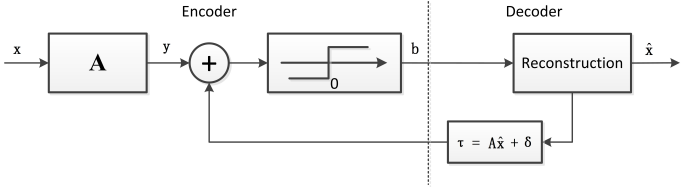}
\caption{Schematic of one-bit adaptive quantization for compressed
sensing.} \label{fig8}
\end{figure}

\subsection{Adaptive Quantization}
As indicated earlier, in addition to the computational issue, the
other difficulty we face in developing a practical algorithm is
that the suggested optimal quantization thresholds are dependent
on the original unquantized data samples $\boldsymbol{y}$ which
are inaccessible to the decoder. To overcome this difficulty, we
consider an adaptive quantization strategy in which the thresholds
are iteratively refined based on the previous
estimate/reconstruction.

The basic idea of one-bit adaptive quantization is described as
follows. At iteration $t$, we compute the estimated measurements
$\boldsymbol{\hat{y}}^{(t)}$ at the decoder based on the sparse
signal $\boldsymbol{\hat{x}}^{(t-1)}$ recovered in the previous
iteration:
$\boldsymbol{\hat{y}}^{(t)}=\boldsymbol{A}\boldsymbol{\hat{x}}^{(t-1)}$.
This estimate is then used to update the quantization thresholds:
$\boldsymbol{\tau}^{(t)}=\boldsymbol{\hat{y}}^{(t)}+\boldsymbol{\delta}^{(t)}$,
where $\boldsymbol{\delta}^{(t)}$ is a vector randomly generated
from a certain distribution. In practice, the deviation
$\boldsymbol{\delta}^{(t)}$ should be gradually decreased to
ensure that the thresholds will eventually come close to
$\boldsymbol{y}$. The updated thresholds $\boldsymbol{\tau}^{(t)}$
are then fed back to the encoder. At the encoder, we compare the
unquantized measurements $\boldsymbol{y}$ with the updated
thresholds $\boldsymbol{\tau}^{(t)}$, and obtain a new set of
one-bit measurements $\boldsymbol{b}^{(t)}$, which are sent to the
decoder. Based on $\boldsymbol{\tau}^{(t)}$ and the new data
$\boldsymbol{b}^{(t)}$, at the decoder, we compute a new estimate
of the sparse signal $\boldsymbol{\hat{x}}^{(t)}$ via solving the
optimization (\ref{opt-2}) or (\ref{opt-3}). A schematic of the
proposed adaptive quantization scheme is shown in Fig. \ref{fig8}.
Note that throughout this iterative process, the unquantized
measurements $\boldsymbol{y}$ are unchanged. For clarity, the
one-bit adaptive quantization scheme is summarized as follows.

\begin{center}
One-bit adaptive quantization scheme
\end{center}
\vspace{-0.2cm} \hrulefill
\begin{enumerate}
\item Given an initial estimate $\boldsymbol{\hat{y}}^{(0)}$, and randomly generate an initial deviation vector $\boldsymbol{\delta}^{(0)}$
according to a certain distribution.
\item At iteration $t\geq 0$, let
$\boldsymbol{\hat{y}}^{(t)}=\boldsymbol{\hat{y}}^{(0)}$ if $t=0$;
otherwise compute $\boldsymbol{\hat{y}}^{(t)}$ as
$\boldsymbol{\hat{y}}^{(t)}=\boldsymbol{A}\boldsymbol{\hat{x}}^{(t-1)}$.
Based on $\boldsymbol{\hat{y}}^{(t)}$, update the thresholds as:
$\boldsymbol{\tau}^{(t)}=\boldsymbol{\hat{y}}^{(t)}+\boldsymbol{\delta}^{(t)}$,
where $\boldsymbol{\delta}^{(t)}$ for $t>0$ is randomly generated
according to a certain distribution with a decreasing variance.
Compare $\boldsymbol{y}$ with the updated thresholds
$\boldsymbol{\tau}^{(t)}$ and obtain a new set of one-bit
measurements $\boldsymbol{b}^{(t)}$.
\item Based on $\boldsymbol{\tau}^{(t)}$ and $\boldsymbol{b}^{(t)}$,
compute a new estimate of the sparse signal
$\boldsymbol{\hat{x}}^{(t)}$ via solving the optimization
(\ref{opt-2}) or (\ref{opt-3}).
\item Go to Step 2 if
$\|\boldsymbol{\hat{x}}^{(t)}-\boldsymbol{\hat{x}}^{(t-1)}\|>\omega$,
where $\omega$ is a prescribed tolerance value; otherwise stop.
\end{enumerate}
\hrulefill

%simplifies the encoding process

%On the other hand, as compared with sending the original
%real-valued measurements to the decoder, the amount of information
%required to be transmitted to the decoder is considerably reduced
%as compared with .

As indicated in \cite{LaskaWen11,BourquardAguet10}, an important
benefit brought by the one-bit design is the significant reduction
of the hardware complexity. One-bit quantizer which takes the form
of a simple comparator is particularly appealing in hardware
implementations, and can operate at a much higher sampling rate
than the high-resolution quantizer. Also, one-bit measurements are
much more amiable for large-scale parallel processing than
high-resolution data. With these merits, the proposed adaptive
architecture enables us to develop data acquisition devices with
lower-cost and faster speed, meanwhile achieving reconstruction
performance similar to that of using multiple-bit quantizer.

The proposed adaptive quantization scheme shares a similar
architecture with the method \cite{KamilovBourguard12} in that
both methods iteratively refine the thresholds using estimates
obtained in previous iteration. Nevertheless, the rationale behind
these two methods are different. Our proposed adaptive
quantization scheme is based on Theorem \ref{theorem1} which
suggests that an arbitrarily small reconstruction error can be
attained if the thresholds are set close enough to the unquantized
measurements $\boldsymbol{y}$, whereas for
\cite{KamilovBourguard12}, a similar theoretical guarantee is
unavailable.

\begin{figure*}[!t]
 \centering
\subfigure[$\ell_1$-minimization]{\includegraphics[width=9cm]{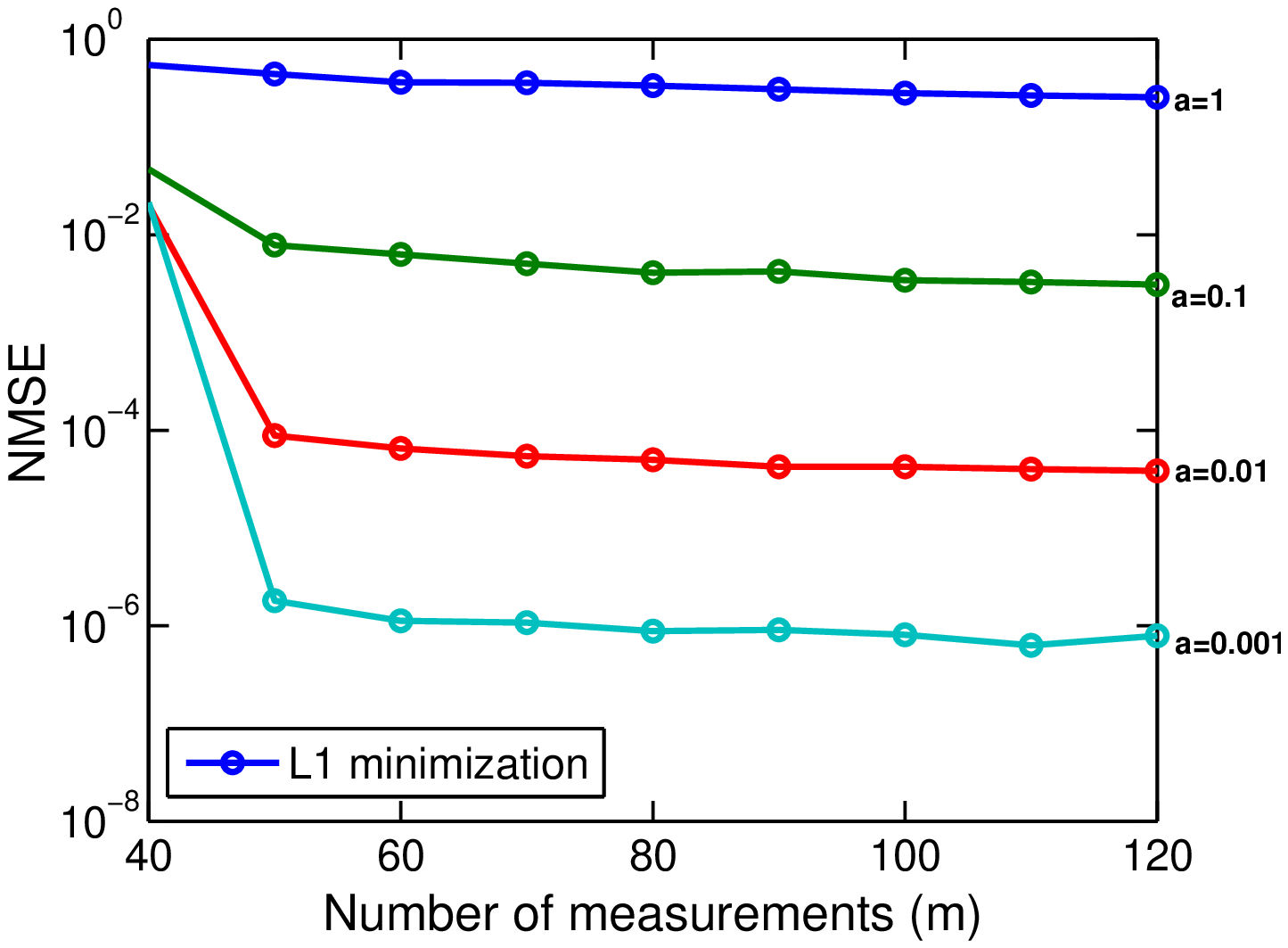}}
 \hfil
\subfigure[Log-sum
minimization]{\includegraphics[width=9cm]{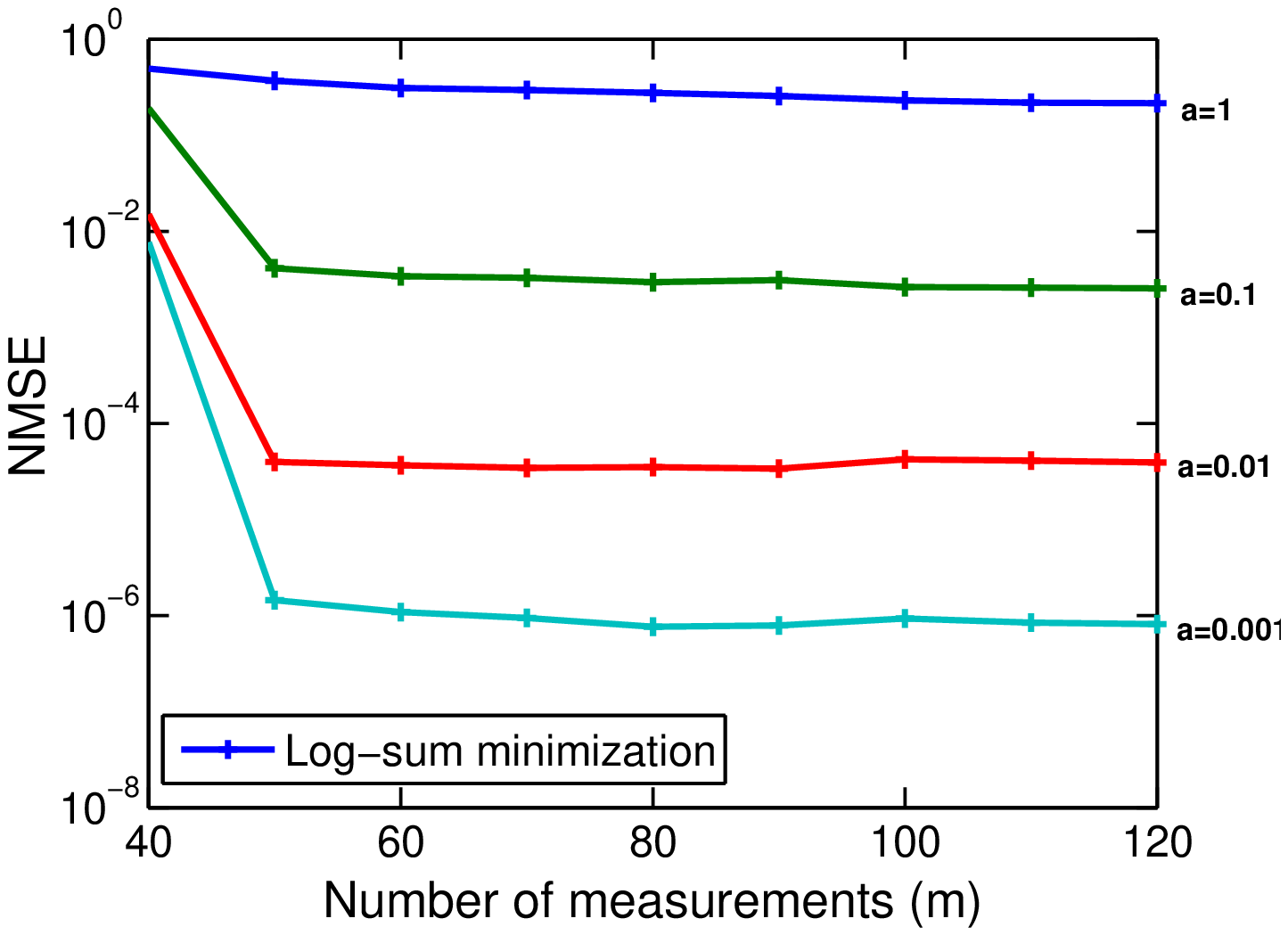}}
  \caption{Reconstruction mean-squared error versus the number of measurements for different choices of $a$.}
   \label{fig1}
\end{figure*}

\section{Numerical Results} \label{sec:numerical-results}
We now carry out experiments to corroborate our previous analysis
and to illustrate the performance of the proposed adaptive
quantization scheme. In our simulations, the $K$-sparse signal is
randomly generated with the support set of the sparse signal
randomly chosen according to a uniform distribution. The signals
on the support set are independent and identically distributed
(i.i.d.) Gaussian random variables with zero mean and unit
variance. The measurement matrix
$\boldsymbol{A}\in\mathbb{R}^{m\times n}$ is randomly generated
with each entry independently drawn from Gaussian distribution
with zero mean and unit variance. To circumvent the difficulty of
solving (\ref{opt-1}), we replace the $\ell_0$ norm with
alternative sparsity-encouraging functionals, namely, $\ell_1$
norm and the log-sum penalty function. The new formulated
optimization problems (c.f. (\ref{opt-2}) and (\ref{opt-3})) can
be efficiently solved.

\subsection{Performance under Different Threshold Choices}
We first examine the impact of the quantization design on the
reconstruction performance. The knowledge of the original
unquantized measurements $\boldsymbol{y}$ is assumed available in
order to validate our theoretical results. The thresholds are
chosen to be the sum of the unquantized measurements
$\boldsymbol{y}$ and a deviation term $\boldsymbol{\delta}$, i.e.
$\boldsymbol{\tau}=\boldsymbol{y}+\boldsymbol{\delta}$, where
$\boldsymbol{\delta}$ is a vector with its entries being
independent discrete random variables with $P(\delta_i=-a)=0.5$
and $P(\delta_i=a)=0.5$, in which the parameter $a>0$ controls the
deviation of $\boldsymbol{\tau}$ from $\boldsymbol{y}$. Fig.
\ref{fig1} depicts the reconstruction normalized mean squared
error (NMSE),
$E[\frac{\|\boldsymbol{x}-\boldsymbol{\hat{x}}\|^2}{\|\boldsymbol{x}\|^2}]$,
vs. the number of measurements $m$ for different choices of $a$,
where we set $n=50$, and $K=3$. Results are averaged over $10^4$
independent runs. From Fig. \ref{fig1}, we see that the
reconstruction performance can be significantly improved by
reducing the deviation parameter $a$. In particular, a NMSE as
small as $10^{-6}$ can be achieved when $a$ is set $0.001$. This
corroborates our theoretical analysis that sparse signals can be
recovered with an arbitrarily small error by letting
$\boldsymbol{\delta}\rightarrow\boldsymbol{0}$. Also, as expected,
the reconstruction error decreases with an increasing number of
measurements $m$. Nevertheless, the performance improvement due to
an increasing $m$ is mild when $m$ is large. This fact suggests
that the choice of quantization thresholds is a more critical
factor than the number of measurements in achieving an accurate
reconstruction. From Fig. \ref{fig1}, we also see that
$\ell_1$-minimization and log-sum minimization provide similar
reconstruction performance. In Fig. \ref{fig2}, we plot the root
mean squared error (RMSE),
$E[\|\boldsymbol{x}-\boldsymbol{\hat{x}}\|_2]$, as a function of
the deviation magnitude
$\epsilon=\|\boldsymbol{\delta}\|_2=\sqrt{m}a$, where we set
$m=100$, $n=120$, $K=2$ and $a$ varies from $10^{-3}$ to $1$. It
can be observed that the RMSE decreases proportionally with the
value $\epsilon$, which coincides with our theoretical analysis
(\ref{theorem1:errorbound}).

To further corroborate our analysis, we consider a different way
to generate the deviation vector $\boldsymbol{\delta}$, with its
entries randomly generated according to a Gaussian distribution
with zero mean and variance $\sigma^2$. Fig. \ref{fig3} depicts
the NMSE vs. the number of measurements $m$ for different values
of $\sigma$. Again, we observe that a more accurate estimate is
achieved when the thresholds get closer to the unquantized
measurements $\boldsymbol{y}$.

Experiments are also carried out on real world images in order to
validate our theoretical results. As it is well-known, images have
sparse structures in certain over-complete basis, such as wavelet
or discrete cosine transform (DCT) basis. In our experiments, we
sample each column of the $256\times 256$ image using a randomly
generated measurement matrix $\boldsymbol{A}\in
\mathbb{R}^{m\times 256}$. We then quantize each real-valued
measurement into one bit of information by using the threshold
$\tau_i=y_i +\delta_i$, in which $\delta_i$ is drawn from a
Gaussian distribution $\mathcal{N}(0, 0.01)$. Fig. \ref{fig4}
shows the original image and the reconstructed images based on
$m\times 256$ one-bit measurements, where $m$ is equal to 150,
200, and 250, respectively. We see that the images restored from
one-bit measurements still provide a decent effect, given that the
thresholds are well-designed.

%In previous simulations, the knowledge of the original unquantized
%measurements $\boldsymbol{y}$ is assumed available in order to
%validate our theoretical results. In practice, only quantized
%measurements are accessible. To overcome this difficulty, an
%one-bit adaptive quantization strategy is proposed in Section
%\ref{sec:algorithm}, where the thresholds are refined iteratively
%based on previous estimate of the sparse signal. We provide
%numerical results to show the effectiveness of the proposed
%iterative algorithm.

%to provide reasonable initial thresholds

\subsection{Performance of Adaptive Quantization Scheme}
We now carry out experiments to illustrate the performance of the
proposed adaptive quantization algorithm. For simplicity, the
algorithm uses the optimization (\ref{opt-2}) at Step 3 of each
iteration. In our experiments, we set $n=50$, $K=2$,
$\boldsymbol{\tau}^{(0)}=\boldsymbol{\delta}^{(0)}$ and the
threshold vector is updated as
$\boldsymbol{\tau}^{(t)}=\boldsymbol{\hat{y}}^{(t)}+\xi^{(t)}\boldsymbol{\delta}^{(t)}$
for $t>0$, where $\boldsymbol{\delta}^{(t)}, \forall t$ is a
random vector with its entries following a normal distribution
with zero mean and unit variance. The parameter $\xi^{(t)}$
controls the magnitude of the deviation error. We set
$\xi^{(0)}=1$, and to gradually decrease the deviation error,
$\xi^{(t)}$ is updated according to $\xi^{(t+1)}=\xi^{(t)}/10$.
The NMSE vs. the number of iterations is plotted in Fig.
\ref{fig5}, where we set $m=40$. Results are averaged over $10^3$
independent runs, with the sampling matrix and the sparse signal
randomly generated for each run. From Fig. \ref{fig5}, we see that
the adaptive algorithm provides a consistent performance
improvement through iteratively refining the quantization
thresholds, and usually provides a reasonable reconstruction
performance within only a few iterations. Fig. \ref{fig6} depicts
the NMSEs as a function of the number of measurements for the
one-bit adaptive quantization scheme and a non-adaptive one-bit
scheme which uses $\boldsymbol{\tau}^{(0)}$ as its thresholds. For
the adaptive scheme, the iterative process stops if
$\|\boldsymbol{\hat{x}}^{(t)}-\boldsymbol{\hat{x}}^{(t-1)}\|_2<0.01$.
Numerical results show that the adaptive scheme usually converges
within ten iterations. We observe from Fig. \ref{fig6} that the
adaptive scheme presents a clear performance advantage over the
non-adaptive method.

%In Fig. \ref{fig7}, we examine the behavior of the adaptive scheme
%when the initial thresholds $\boldsymbol{\tau}^{(0)}$ are randomly
%generated, i.e
%$\boldsymbol{\tau}^{(0)}=\xi^{(0)}\boldsymbol{\delta}$. The NMSEs
%of the adaptive and non-adaptive schemes vs. the number of
%measurements are plotted in Fig. \ref{fig7}, which shows that a
%performance gain can still be obtained when the initial thresholds
%are arbitrarily chosen.

We examine the effectiveness of the adaptive quantization scheme
for image recovery. In the experiments, we sample each column of
the $128\times 128$ image using a randomly generated measurement
matrix $\boldsymbol{A}\in \mathbb{R}^{m\times 128}$, and then
quantize each real-valued measurement into one bit of information.
The initial threshold vector is set to be
$\boldsymbol{\tau}^{(0)}=\boldsymbol{y}+\xi^{(0)}\boldsymbol{\delta}^{(0)}$,
where $\xi^{(0)}$ is chosen as $\|\boldsymbol{y}\|_l/m$ such that
the deviation is comparable to the magnitude of entries in
$\boldsymbol{y}$. For $t>0$, the threshold is then updated as
$\boldsymbol{\tau}^{(t)}=\boldsymbol{\hat{y}}^{(t)}+\xi^{(t)}\boldsymbol{\delta}^{(t)}$,
with $\xi^{(t)}=\xi^{(t-1)}/10$. Fig. \ref{fig8} plots the images
reconstructed by the proposed one-bit adaptive scheme and the
non-adaptive one-bit scheme which uses $\boldsymbol{\tau}^{(0)}$
as its thresholds, where we set $m=256$. Fig. \ref{fig8}
demonstrates that the adaptive quantization scheme improves the
reconstruction of the image significantly.

\begin{figure}[!t]
\centering
\includegraphics[width=9cm]{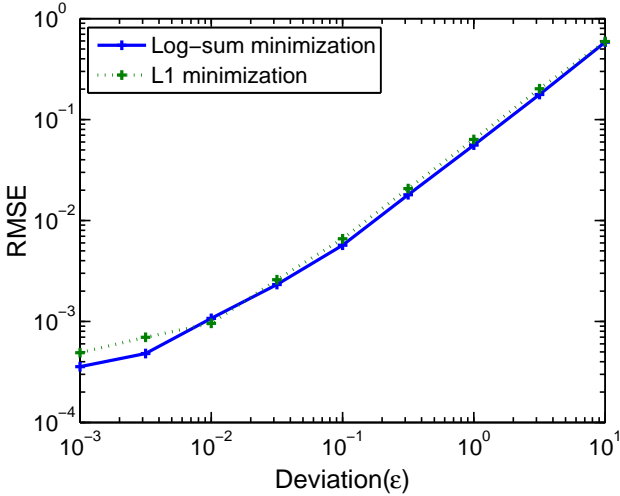}
\caption{Root mean-squared error versus $\epsilon$.} \label{fig2}
\end{figure}

\begin{figure*}[!t]
 \centering
\subfigure[$\ell_1$-minimization]{\includegraphics[width=9cm]{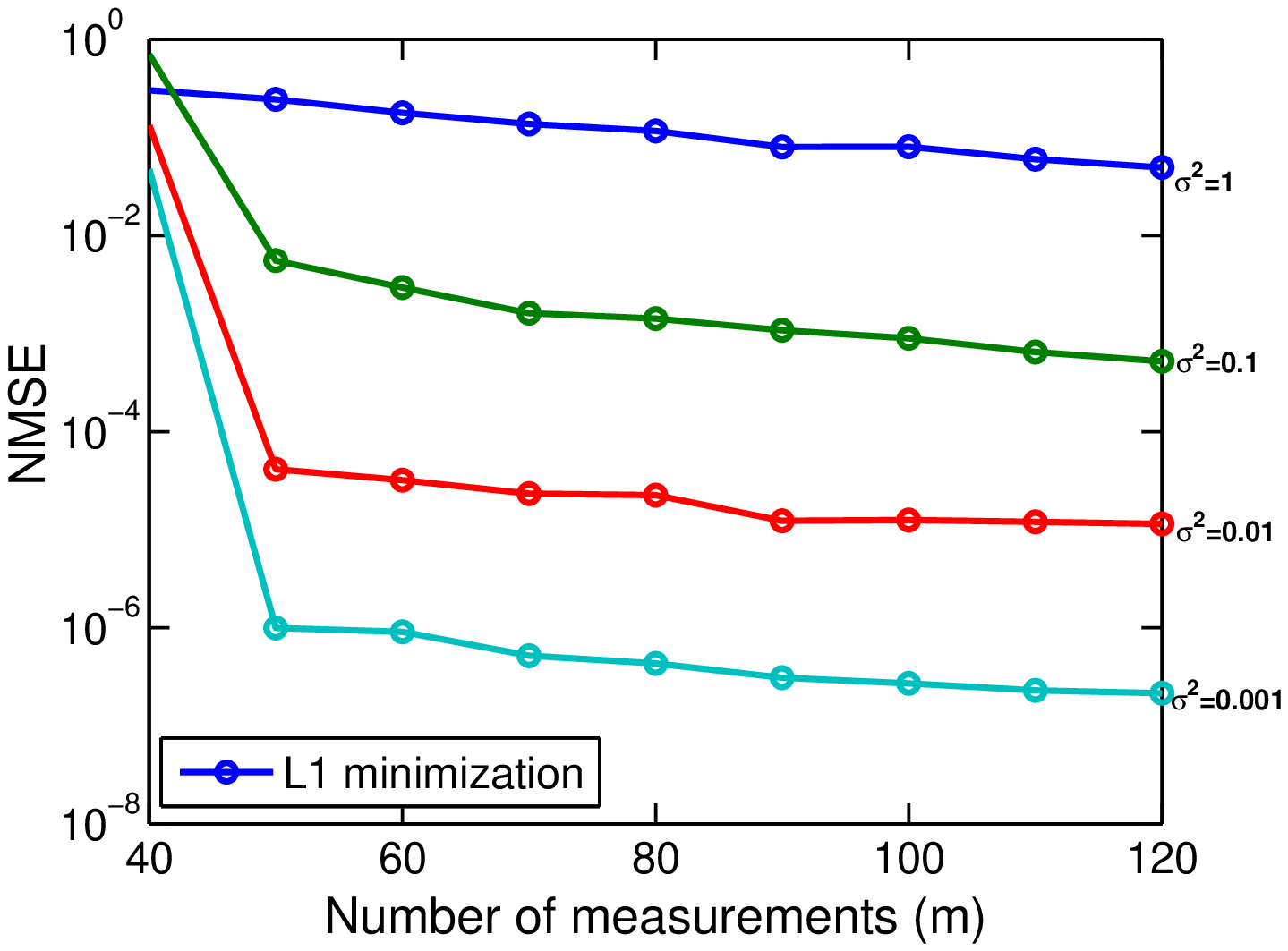}}
 \hfil
\subfigure[Log-sum
minimization]{\includegraphics[width=9cm]{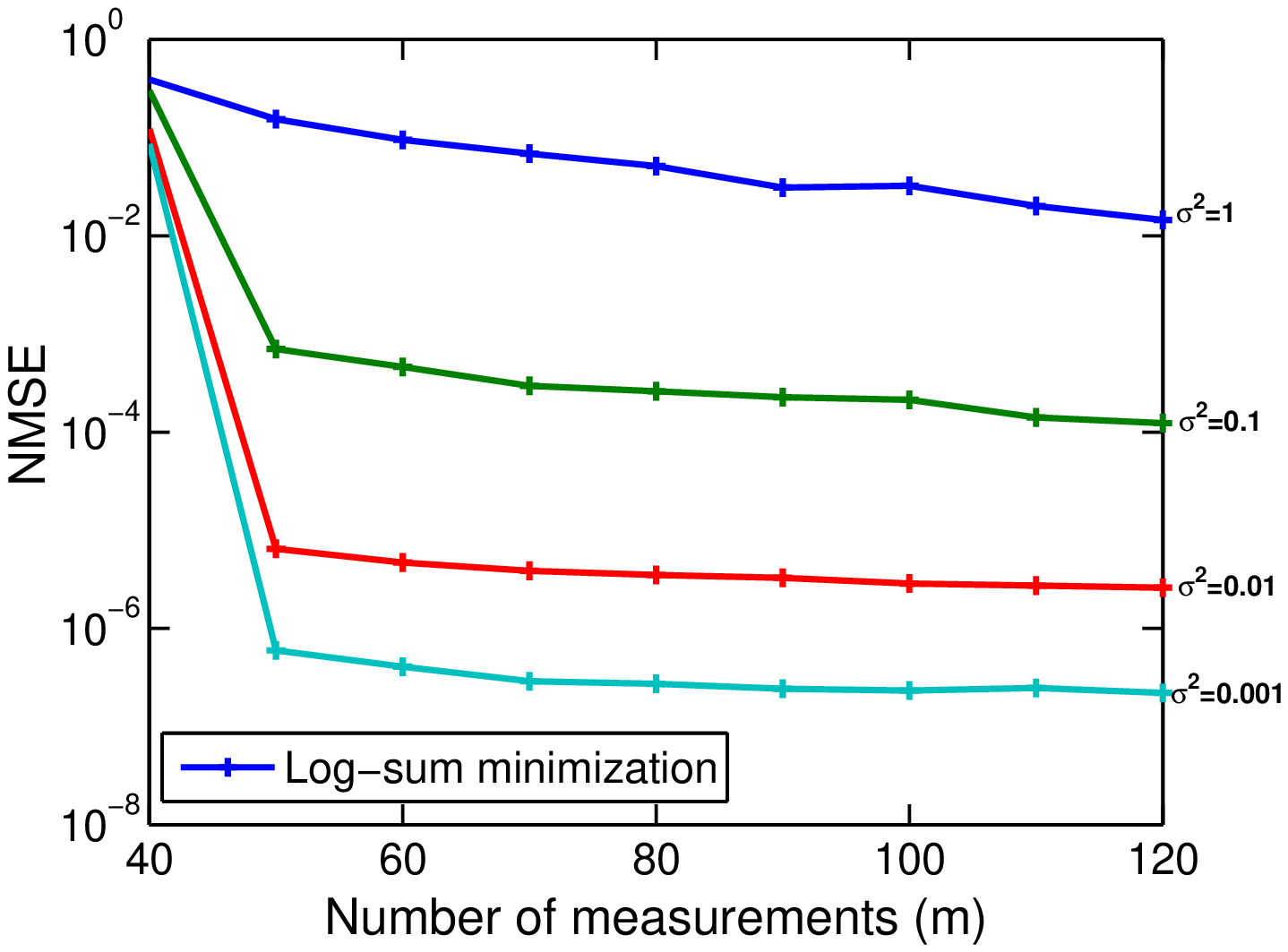}}
  \caption{Reconstruction mean-squared error versus the number of measurements for different choices of $\sigma$.}
   \label{fig3}
\end{figure*}

\begin{figure*}[!t]
 \centering
\begin{tabular}{cccc}
\hspace*{-3ex}
\includegraphics[width=4.8cm]{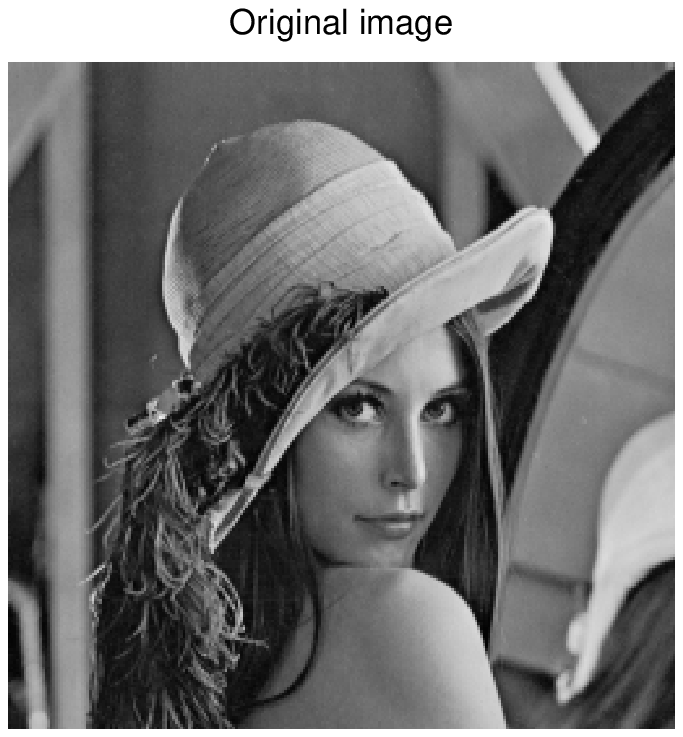}&
\hspace*{-5ex}
\includegraphics[width=4.8cm]{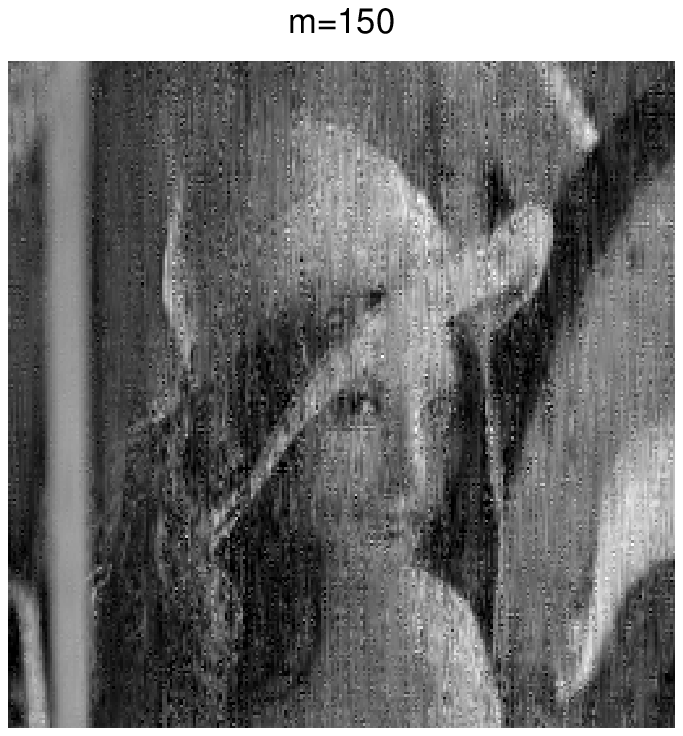}
\hspace*{-3ex}
\includegraphics[width=4.8cm]{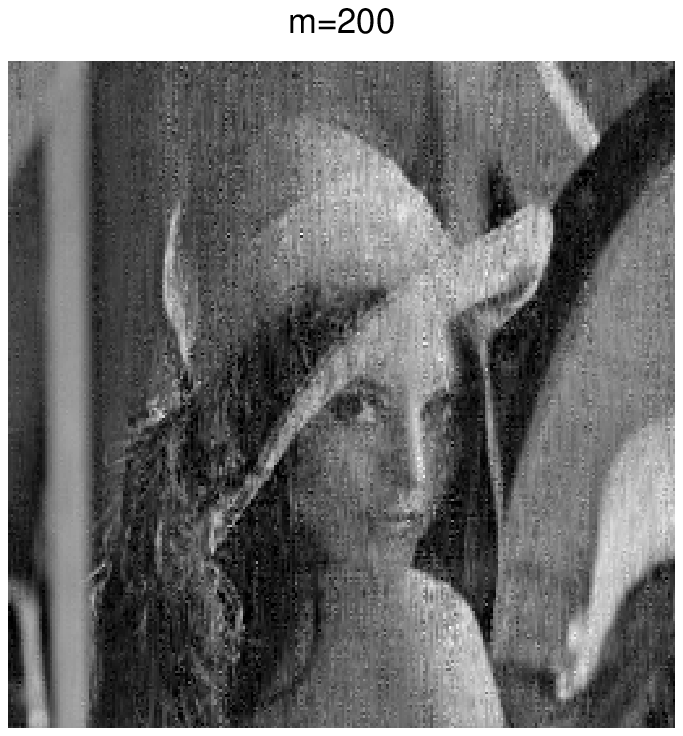}&
\hspace*{-5ex}
\includegraphics[width=4.8cm]{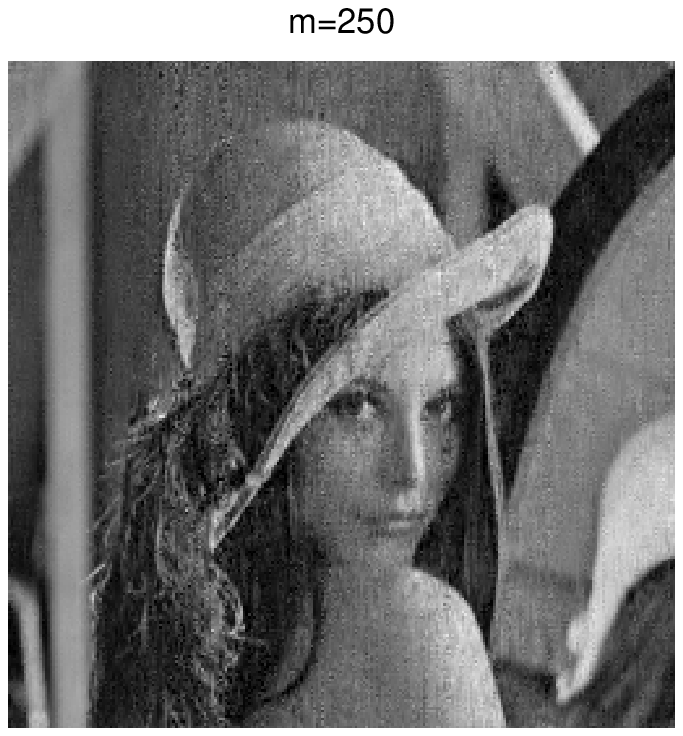}
\end{tabular}
  \caption{Original image and reconstructed images based on $m\times 256$ one-bit measurements.} \label{fig4}
\end{figure*}

\begin{figure}[!t]
\centering
\includegraphics[width=9cm]{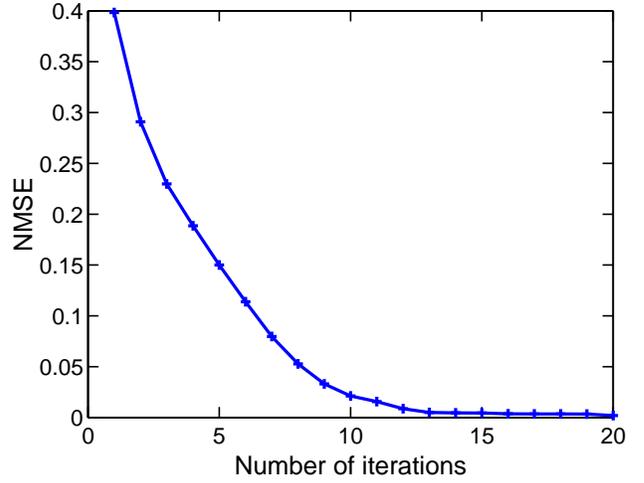}
\caption{Normalized mean-squared error versus number of
iterations.} \label{fig5}
\end{figure}

%\begin{figure}[!t]
%\centering
%\includegraphics[width=9cm]{AQmsevsm}
%\caption{Normalized mean-squared errors versus the number of
%measurements.} \label{fig6}
%\end{figure}

\begin{figure}[!t]
\centering
\includegraphics[width=9cm]{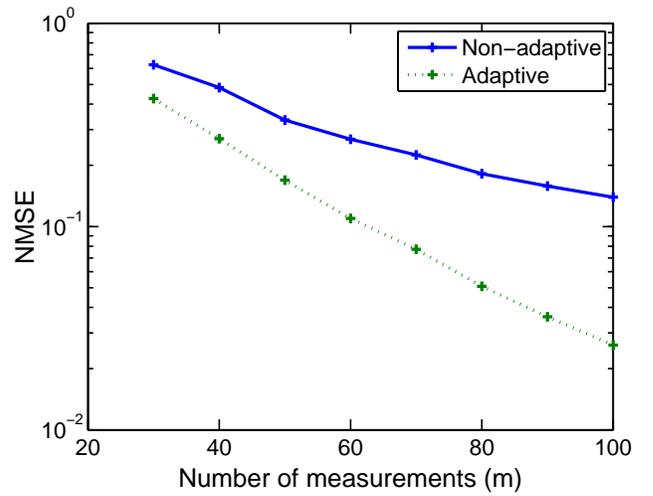}
\caption{Normalized mean-squared errors versus the number of
measurements.} \label{fig6}
\end{figure}

\begin{figure}[!t]
 \centering
\begin{tabular}{ccc}
\hspace*{-3ex}
\includegraphics[width=4.9cm]{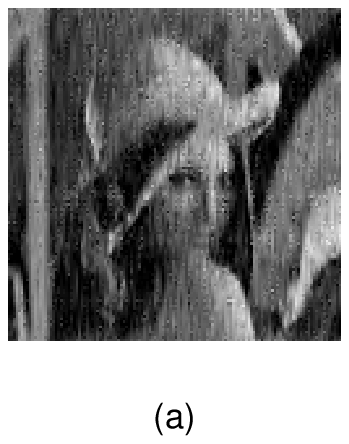}&
\hspace*{-3ex}
\includegraphics[width=4.9cm]{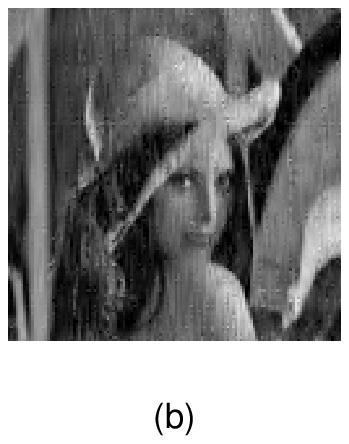}
\hspace*{-3ex}
\end{tabular}
  \caption{Reconstruction of images from one-bit measurements: (a)
Image reconstructed by non-adaptive one-bit scheme; (b) Image
reconstructed by adaptive one-bit scheme.} \label{fig8}
\end{figure}

%\begin{figure*}[!t]
% \centering
%\subfigure[Non-adaptive one-bit
%scheme]{\includegraphics[width=9cm]{image_NQ}}
% \hfil
%\subfigure[Adaptive one-bit
%scheme]{\includegraphics[width=9cm]{image_AQ}}
%  \caption{Reconstruction of images from one-bit measurements.}
%   \label{fig8}
%\end{figure*}

%In which, we examined the performance of quantization using
%constant number of bits $q_1=q_2=q_3=2$ in each iteration,
%referred to as ``constant bits'', or alternatively applying
%decreasing number of bits $q_1=4$, $q_2=3$, $q_3=2$ respectively
%in each iteration, referred to as ``decreasing bits''. For both
%optimization techniques and quantization schemes, desirable
%performance is obtained. Moreover, it can been seen that the
%log-sum optimization with constant quantization bits outperforms
%the other strategies, such result could give some guidance for
%practical implantations.

%Fig. \ref{fig7} illustrates the performance of the two
%optimization algorithms in each iteration. The parameters are set
%at $m=10$, $n=15$, and $K=3$. The number of bits applied in each
%iteration is $q_i=3, \forall i=1,2,3,4$. Both algorithms converge
%fast and sound reconstruction performance is obtained.

\section{Conclusion} \label{sec:conclusion}
We studied the problem of one-bit quantization design for
compressed sensing. Specifically, the following two questions were
addressed: how to choose quantization thresholds, and how close
can the reconstructed signal be to the original signal when the
quantization thresholds are well-designed? Our analysis revealed
that sparse signals can be recovered with an arbitrarily small
error by setting the thresholds close enough to the original
unquantized measurements. The unquantized measurements,
unfortunately, are inaccessible to us. To address this issue, we
proposed an adaptive quantization method which iteratively refines
the quantization thresholds based on previous estimate. Simulation
results were provided to corroborate our theoretical analysis and
to illustrate the effectiveness of the proposed adaptive
quantization scheme.

\bibliography{newbib}
\bibliographystyle{IEEEtran}

\end{document}